\documentclass{article}

\usepackage{amsthm}
\usepackage{amssymb}
\usepackage{amsfonts}
\usepackage{algorithm}
\usepackage{algpseudocode}
\usepackage{amsmath, amssymb, amsthm}
\usepackage{graphicx}
\usepackage{hyperref}
\usepackage{url}
\usepackage{geometry}
\usepackage{enumitem}
\usepackage{authblk} 

\usepackage{color}
\usepackage{comment}
\usepackage{bm}
\usepackage{url}
\usepackage{thm-restate}
\usepackage{booktabs}
\usepackage{ulem}
\usepackage{here}
\usepackage{placeins}

\newcommand*{\myblock}[1]{\subparagraph{#1.}}

\newcommand{\calX}{{\mathcal{X}}}
\newcommand{\calF}{{\mathcal{F}}}
\newcommand{\calFcopy}{{\mathcal{F}_{\rm copy}}}

\newcommand{\srcL}[1]{{{\rm src}_{\rm L}({#1})}}
\newcommand{\srcR}[1]{{{\rm src}_{\rm R}({#1})}}
\newcommand{\posL}[1]{{{\rm pos}_{\rm L}({#1})}}
\newcommand{\posR}[1]{{{\rm pos}_{\rm R}({#1})}}

\newcommand{\rel}[1]{{{\rm rel}({#1})}}
\newcommand{\abs}[1]{{{\rm abs}({#1})}}

\newcommand{\jump}[1]{{{\rm jump}({#1})}}

\newcommand{\calG}{{\mathcal{G}}}

\newcommand{\calD}{{\mathcal{D}}}

\renewcommand{\exp}[1]{{{\rm exp}({#1})}}

\newcommand{\zbe}{{z_{\rm be}}}
\newcommand{\zbeg}{{z_{\rm be}^{\rm g}}}
\newcommand{\zbeopt}{{z_{\rm be}^{\rm OPT}}}

\newcommand{\gopt}{{g^{\rm OPT}}}
\newcommand{\grlopt}{{g_{\rm rl}^{\rm OPT}}}

\newcommand{\head}[1]{{{\rm head}({#1})}}
\newcommand{\tail}[1]{{{\rm tail}({#1})}}
\newcommand{\val}[1]{{{\rm val}({#1})}}

\newcommand{\lca}[1]{{{\rm lca}({#1})}}

\newcommand{\calI}{{\mathcal{I}}}
\newcommand{\calB}{{\mathcal{B}}}

\newcommand{\ns}[1]{{n_{s}({#1})}}
\renewcommand{\ne}[1]{{n_{e}({#1})}}

\title{
LZBE: an LZ-style compressor supporting $O(\log n)$-time random access
} %

\usepackage{amsthm}
\newtheorem{theorem}{Theorem}
\newtheorem{lemma}[theorem]{Lemma}
\newtheorem{definition}[theorem]{Definition}
\newtheorem{corollary}[theorem]{Corollary}

\author[1]{Hiroki Shibata\thanks{\texttt{shibata.hiroki.753@s.kyushu-u.ac.jp}}}
\author[2]{Yuto Nakashima\thanks{\texttt{nakashima.yuto.003@m.kyushu-u.ac.jp}}}
\author[3]{Yutaro Yamaguchi\thanks{\texttt{yutaro.yamaguchi@ist.osaka-u.ac.jp}}}
\author[2]{Shunsuke Inenaga\thanks{\texttt{inenga.shunsuke.380@m.kyushu-u.ac.jp}}}

\affil[1]{Joint Graduate School of Mathematics for Innovation, Kyushu University, Japan}
\affil[2]{Department of Informatics, Kyushu University, Kyushu University, Japan}
\affil[3]{Department of Information and Physical Sciences, Osaka University, Japan}

\begin{document}

\maketitle

\begin{abstract}
An LZ-like factorization of a string divides it into factors, each being either a single character or a copy of a preceding substring.
While grammar-based compression schemes support efficient random access with space linear in the compressed size, no comparable guarantees are known for general LZ-like factorizations.  
This limitation motivated restricted variants such as LZ-End [Kreft and Navarro, 2013] and height-bounded LZ (LZHB) [Bannai et al., 2024], which trade off some compression efficiency for faster access.
In this paper, we introduce \emph{LZ-Begin-End} (\emph{LZBE}), a new LZ-like variant in which every copy factor must refer to a contiguous sequence of preceding factors.
This structural restriction ensures that any context-free grammar can be transformed into an LZBE factorization of the same size.
We further study the greedy LZBE factorization, which selects each copy factor to be as long as possible while processing the input from left to right, and show that it can be computed in linear time.
Moreover, we exhibit a family of strings for which the greedy LZBE factorization is asymptotically smaller than the smallest grammar.
These results demonstrate that the LZBE scheme is strictly more expressive than grammar-based compression in the worst case.
To support fast queries, we propose a data structure for LZBE-compressed strings that permits $O(\log n)$-time random access within space linear in the compressed size, where $n$ is the length of the input string.
\end{abstract}

\section{Introduction}

Data compression is a fundamental technique for reducing storage and transmission costs in large-scale data processing.  
Its importance is particularly evident for \emph{highly repetitive} string collections, which share many long repeated patterns.
Such data naturally arise in various domains, including version-controlled data and genomic databases, and exhibit substantial redundancy that can be effectively exploited by modern compression algorithms~\cite{DBLP:journals/csur/Navarro21a}.

Computing directly on compressed data without decompressing it into its original form is a fundamental task in data compression and compressed data structures.  
Among those studied in this context, \emph{random access} to retrieve any character of the original text on its compressed representation is one of the most fundamental and widely studied problems~\cite{DBLP:journals/siamcomp/BilleLRSSW15,DBLP:conf/soda/KempaS22}.  
Efficient random access is a core primitive for various higher-level operations in compressed space and thus plays a central role in the design of compressed data structures.

A \emph{Lempel-Ziv (LZ)-like factorization} is a factorization of a string in which each factor is either a single character or a copy of an earlier substring of the same text.  
First introduced nearly fifty years ago~\cite{DBLP:journals/tit/LempelZ76},  
LZ-like factorization has remained a central framework of both theory and practice,  
powering classic tools such as \texttt{gzip} and still forming a core component of modern compressors like \texttt{zstd} and \texttt{brotli}~\cite{DBLP:journals/rfc/rfc7932,DBLP:journals/rfc/rfc8478,DBLP:journals/rfc/rfc1952}.

In general, stronger compression tends to introduce more complex dependency structures, making efficient queries on compressed data more difficult.
This trade-off is well illustrated by the contrast between grammar-based and LZ-like compression schemes.
A grammar-based compression is known to be asymptotically at least as large as the corresponding LZ factorization for every string~\cite{DBLP:journals/tcs/Rytter03},
but it allows for random access in $O(\log n)$ time using linear space in the size of the compressed representation~\cite{DBLP:journals/siamcomp/BilleLRSSW15},
where $n$ denotes the length of the input text.  
In contrast, despite achieving stronger compression, no known data structure supports polylogarithmic time random access for general LZ-like factorizations.  

These observations naturally lead to a fundamental theoretical question:  
what is the smallest class of compression schemes that supports efficient random access, ideally in $O(\log n)$ time?  
To address this question, recent studies have explored restricted variants of LZ-like factorizations that impose structural constraints on copy relationships,  
aiming to retain much of the compression efficiency of LZ while allowing efficient random access~\cite{DBLP:conf/esa/BannaiFHMP24,DBLP:journals/tcs/KreftN13,DBLP:conf/cpm/LiptakM024}.

\myblock{Contributions}
This paper proposes the \emph{LZ-Begin-End (LZBE) factorizations},  
a new class of LZ-like factorizations in which each copy factor refers to a contiguous sequence of preceding factors,  
i.e., both the begin and end positions of the reference must align with factor boundaries.  
LZBE lies between grammar-based compression and LZ-Begin/LZ-End in terms of structure and compression performance.  
Our contributions are summarized as follows:
\begin{enumerate}
    \item 
    We clarify the theoretical relationship between LZBE factorizations and grammar-based compression.  
    First, we show that any context-free grammar can be converted into an LZBE factorization without increasing its size.
    We then prove that the greedy LZBE factorization can be asymptotically smaller than the smallest grammar, and establish a matching upper bound on the grammar–LZBE size ratio.
    Together, these results establish the relationship between the two representations with matching upper and lower bounds.
    \item We propose a linear-space compressed data structure that supports random access in $O(\log n)$ time.
    This establishes LZBE as one of the smallest known compression schemes that support $O(\log n)$-time random access using space linear in the compressed size.
    \item We present a linear-time algorithm for computing the greedy LZBE factorization.
\end{enumerate}

\myblock{Related Works}
To address the limitations of traditional LZ-like factorizations, several restricted variants have been proposed, including LZ-Begin/LZ-End~\cite{DBLP:journals/tcs/KreftN13} and Height-Bounded LZ (LZHB, also known as BAT-LZ)~\cite{DBLP:conf/esa/BannaiFHMP24, DBLP:conf/cpm/LiptakM024}.  
These variants aim to support efficient random access while keeping space usage linear in their compressed size.  
LZ-Begin/LZ-End restricts each copy factor to reference a previous substring that begins/ends at a factor boundary, whereas LZHB constrains the height of the dependency graph formed by copy relationships.

The LZ-End index~\cite{DBLP:conf/soda/KempaS22} allows for random access in \(O(\log^4 n \log \log n)\) time on LZ-End compressed strings.  
However, that work establishes only the ``existence'' of such a data structure and does not provide an efficient construction algorithm.
The authors note that the construction achieving the worst-case polylogarithmic query time is only known to run in expected polynomial time and no deterministic variant is known.
Improving either the construction time or derandomization remains an open problem.

In contrast, LZHB with a maximum height of \(O({\rm poly}\log n)\) naturally supports \(O({\rm poly}\log n)\)-time random access.
Nevertheless, the theoretical understanding of LZHB remains limited.  
Given an arbitrary LZ-like factorization (even when restricted to LZ-Begin/LZ-End), how to “balance’’ the factorization without increasing its size is still an open problem.  
Moreover, it is known that finding the smallest parsing under a given height bound is NP- and APX-hard~\cite{DBLP:journals/iandc/CicaleseU25}.
Therefore, it remains unclear whether LZHB can outperform other LZ-like factorization methods.

\section{Preliminaries}

Let $\Sigma$ be a set of symbols called the \emph{alphabet}, and $\Sigma^*$ the set of strings.
An element of $\Sigma$ is called a \emph{character}.  
We denote a string $T$ of length $n$ as $T = T[1] \cdots T[n]$, where $T[i]$ is the $i$-th character of $T$.  
We denote the substring from position $i$ to position $j$ in $T$ by $T[i, j]$, for $1 \leq i \leq j \leq n$.
A sequence $\calF = (F_1, \dots, F_f)$ of non-empty strings such that $T = F_1 \cdots F_f$ is called a \emph{factorization} of $T$.
Each element $F_i \in \calF$ is called a \emph{factor}.
Throughout this paper, we use $n$ to denote the length of $T$.

A \emph{context-free grammar (CFG)} generating a single string consists of the set  $\Sigma$ of terminals, the set $\mathcal{X}$ of nonterminals, the set of production rules where each $X \in \mathcal{X}$ has exactly one rule, and the start symbol $S \in \mathcal{X}$.  
We assume the derivation relation is acyclic, so that each nonterminal $X$ expands to a unique string $\exp{X}$.  
The size of the grammar is the total number of symbols on the right-hand sides of all production rules.  
Let $\gopt$ denote the size of the smallest CFG generating a given string $T$.
A \emph{straight-line program (SLP)} is a CFG in Chomsky normal form, where each rule is either $X \rightarrow c$ for $c \in \Sigma$, or $X \rightarrow YZ$ for $Y, Z \in \mathcal{X}$.  
Any CFG of size $g$ can be converted into an equivalent SLP of size $O(g)$.

A \emph{Lempel--Ziv (LZ)-like factorization}~\cite{DBLP:journals/tit/ZivL77} is a factorization in which each factor is either a single character or a substring that occurs earlier in the string.
It is known that the LZ-like factorization obtained by scanning the string left to right and greedily selecting the longest valid factor at each step minimizes the number of factors.
For an LZ-like factorization $\calF$, we refer to factors that are the single character of the first occurrence as \emph{char factors}, and those that copy a previous occurrence as \emph{copy factors}.  
Let $\calFcopy \subseteq \calF$ denote the set of all copy factors.
For any factor $F_i \in \calF$, define $\posL{F_i} = \sum_{j=1}^{i-1} |F_j| + 1$ and $\posR{F_i} = \sum_{j=1}^{i} |F_j|$.  
For each copy factor $F_i \in \calFcopy$, we fix one of its earlier occurrences in $F_1 \cdots F_i$ and define $\srcL{F_i}$ and $\srcR{F_i}$ as its beginning and ending positions in $T$.
Given a position $1 \leq p \leq n$, we define $\rel{p}$ as the pair 
$(F, r)$ where $F$ is the factor containing the position $p$ and $r = p - \posL{F} + 1$ is the relative position of $p$ within $F$.
Conversely, given a factor $F$ and a relative position $1 \leq r \leq |F|$, the absolute position in $T$ is given by $\abs{F, r} = \posL{F} + r - 1$.

Since every copy factor refers to an earlier occurrence, the original string can be restored using only an LZ-like factorization.  
We define the \emph{jump function} $\jump{F, r}$ for a copy factor $F \in \calFcopy$ and position $1 \leq r \leq |F|$ as $\jump{F, r} = \rel{q}$, where $q = \srcL{F} + r - 1$ is the position in $T$ that $F[r]$ refers to.
Given a pair $(F, r)$, the \emph{jump sequence} is defined as the sequence $(F_{i_0}, r_0), (F_{i_1}, r_1), \dots, (F_{i_k}, r_k)$ such that $(F_{i_0}, r_0) = (F, r)$ and $(F_{i_{t+1}}, r_{t+1}) = \jump{F_{i_{t}}, r_{t}}$ for all $t < k$, where $F_{i_k}$ is a char factor representing $T[\abs{F, r}]$.
This sequence traces the chain of references back to the character’s first occurrence.
To restore a character $T[p]$, we follow the jump sequence from $\rel{p}$ and return the character in the final pair.

A \emph{dependency DAG} $\calD = \langle \calF, E \rangle$ for an LZ-like factorization $\calF$  
is a directed acyclic graph where $E = \{ (F_i, F_j) \in \calF \times \calF \mid [\posL{F_j}, \posR{F_j}] \subseteq [\srcL{F_i}, \srcR{F_i}] \}$.  
The outgoing edges from each node $F_i$ are ordered to follow the left-to-right order of the factors.
In this DAG, a jump to a source factor corresponds to traversing an edge, and a jump sequence corresponds to a path to a vertex representing a char factor.

\section{LZBE Factorization}

In this section, we introduce a new factorization scheme called LZ-Begin-End (LZBE) and explore its relationships with other compression schemes.

An \emph{LZ-Begin-End (LZBE) factorization} of a string $T$ is an LZ-like factorization 
$\mathcal{F} = (F_1, \dots, F_\zbe)$ in which 
every copy factor must have a previous occurrence that begins and ends at factor boundaries.
Hence, each copy factor corresponds to a concatenation of a contiguous sequence of preceding factors.
Formally, for every copy factor $F_i$, there exist integers $1 \le j \le k < i$ such that $F_i = F_j F_{j+1} \cdots F_k$.
The \emph{greedy LZBE factorization} of a string $T$ is defined as the LZBE factorization obtained by scanning $T$ from left to right and, at each step, selecting the longest valid factor that satisfies the LZBE condition.  
We note that the greedy LZBE factorization is not necessarily the smallest LZBE factorization (see Appendix~\ref{app:greedy_approx}).
Let $\zbeg$ and $\zbeopt$ denote the number of factors in the greedy and the smallest LZBE factorizations of $T$, respectively.

Since each copy factor refers to a contiguous sequence of preceding factors, an LZBE factorization naturally defines a context-free grammar.  
Specifically, for each copy factor $F_i \in \calFcopy$ of the form $F_i = F_j \cdots F_k$, we define the rule $F_i \rightarrow F_j \cdots F_k$, and for each char factor $F_i$ representing $c$, we define the rule $F_i \rightarrow c$.
The start symbol $S$ is defined by the rule $S \rightarrow F_1 \cdots F_\zbe$.
This grammar generates $T$, and its size is $O(\zbe^2)$, since each copy factor may refer to up to $O(\zbe)$ factors.  
In this sense, the LZBE factorization can be seen as a compact representation of a CFG with a specific form, where each rule produces either a single character or a contiguous sequence of existing nonterminals.
Figure~\ref{fig:LZBE_to_grammar} shows an LZBE factorization and its corresponding grammar.

\begin{figure}[htbp]
    \centering
    \begin{minipage}[b]{0.7\linewidth}
        \centering
        \includegraphics[width=\linewidth]{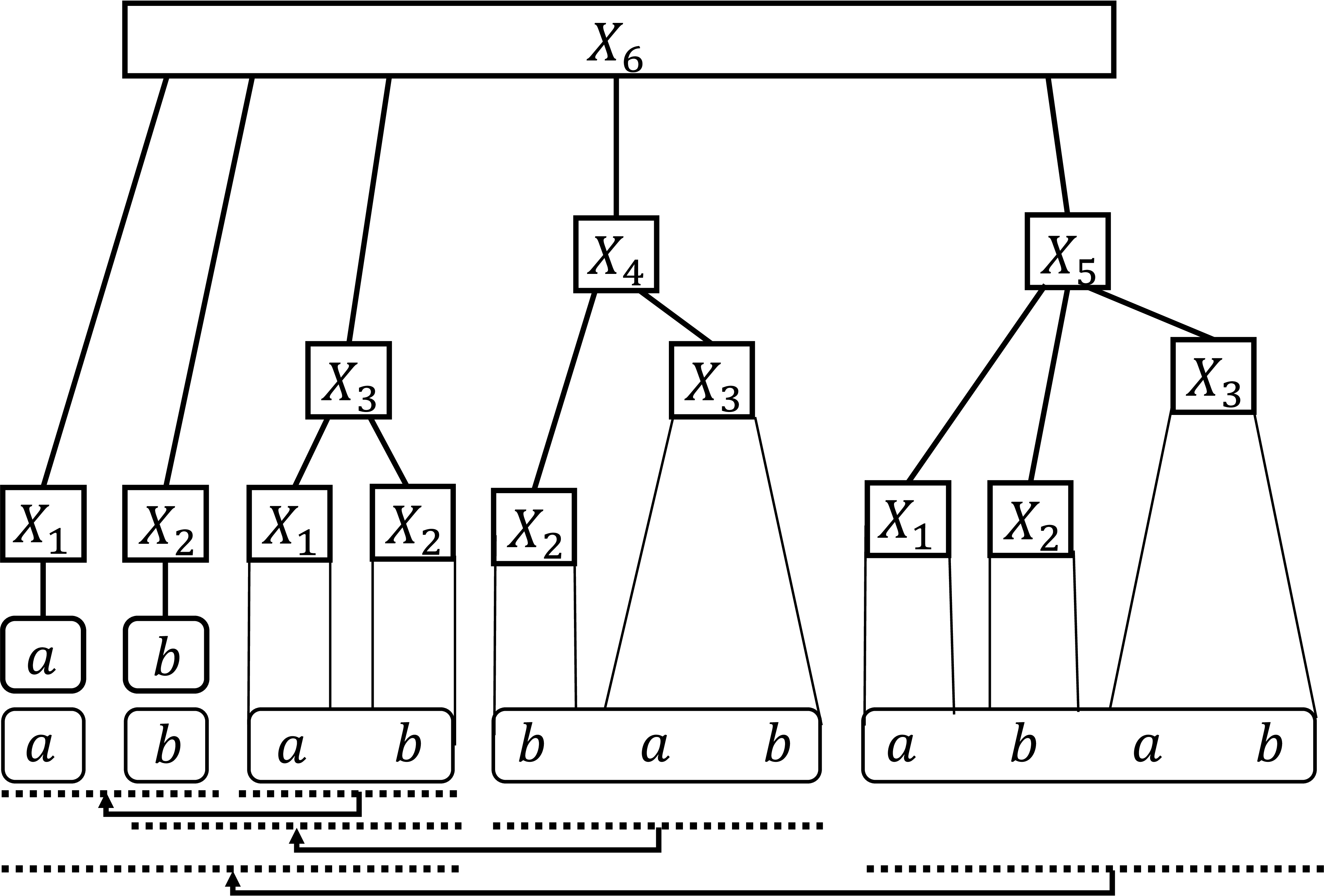}
    \end{minipage}
    \caption{
    The greedy LZBE factorization of $ababbababab$ and the pruned derivation tree of the corresponding CFG.
    Dotted lines and arrows indicate the sources of copy factors.
    }
    \label{fig:LZBE_to_grammar}
\end{figure}

As is shown above, an LZBE factorization can be viewed as a compact representation of a CFG with a specific structure.  
Despite this restricted form, any CFG can be represented as an LZBE factorization without increasing its size.
\begin{theorem}
Given a CFG of size $g$ representing string $T$, we can construct an LZBE factorization with at most $g$ factors.
\end{theorem}
\begin{proof}
We apply the grammar decomposition~\cite{DBLP:journals/tcs/Rytter03} to $T$, which prunes all subtrees rooted at non-leftmost occurrences of nonterminals in the derivation tree.  
Each leaf in the pruned tree becomes a factor: terminals form char factors, and pruned nonterminals form copy factors referring to their leftmost occurrence.  
Since each nonterminal appears at most once as an internal node, the number of factors is bounded by $g$. 
As each copy factor refers to the concatenation of leaves (i.e., factors) in a subtree, it refers to a consecutive preceding factors.  
Hence, the resulting factorization is an LZBE factorization.
\end{proof}
This theorem immediately implies that $\zbeopt \leq \gopt$.
Figure~\ref{fig:grammar_decomposition} shows an example of a grammar and its corresponding LZBE factorization.

\begin{figure}[htbp]
    \centering
    \begin{minipage}[b]{0.48\linewidth}
        \centering
        \includegraphics[width=\linewidth]{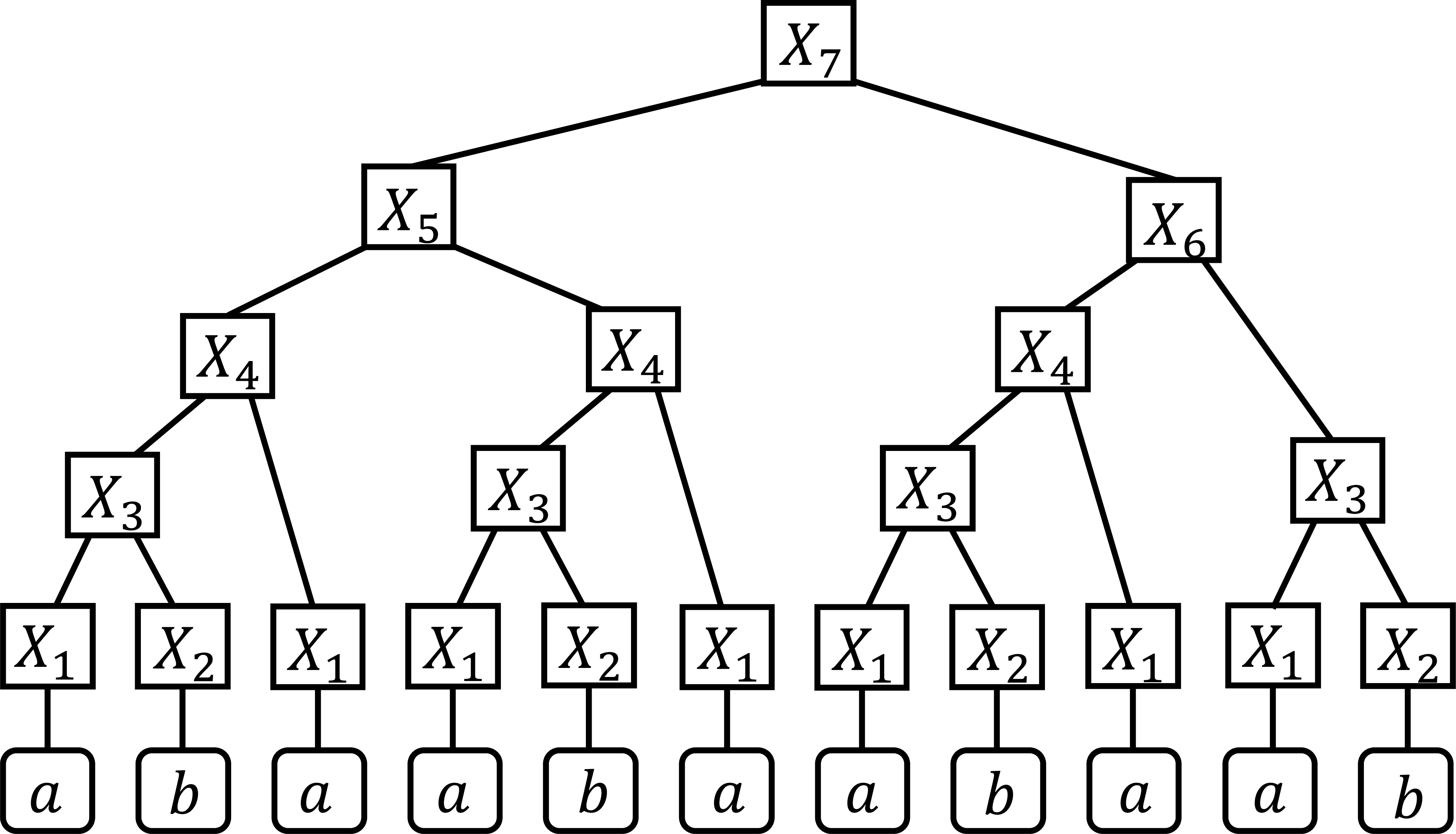}
    \end{minipage}
    \hfill
    \begin{minipage}[b]{0.48\linewidth}
        \centering
        \includegraphics[width=\linewidth]{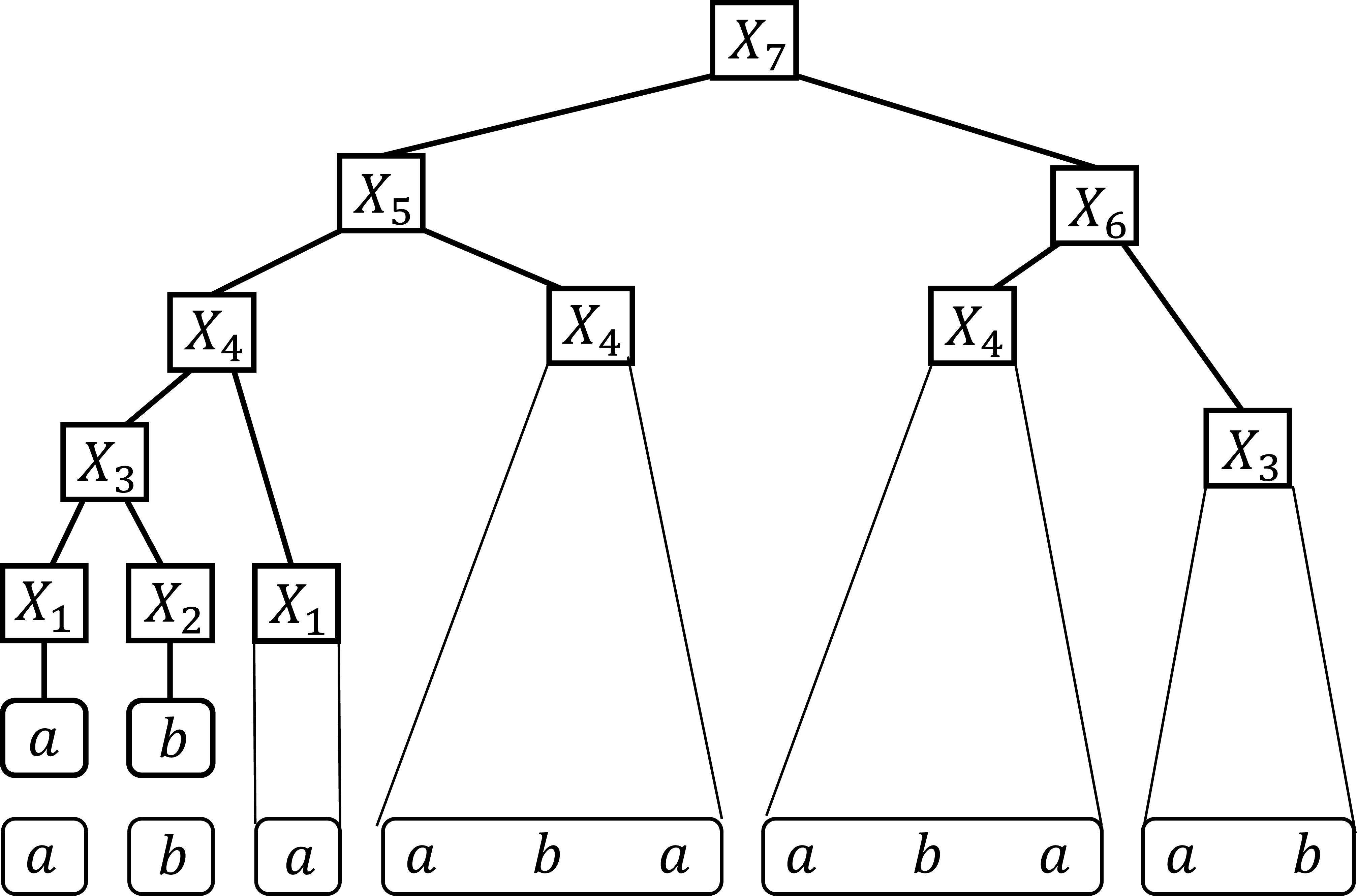}
    \end{minipage}
    \caption{
    The derivation tree of a grammar generating $abaabaabaab$ (left) and its pruned derivation tree (right).
    The LZBE factorization corresponding to the grammar is shown below the pruned tree.
    }
    \label{fig:grammar_decomposition}
\end{figure}

The above discussion shows that the LZBE scheme is at least as powerful as grammar compression.  
A natural question is whether they are asymptotically equivalent or strictly stronger.
To address this, we present the following theorem.

\begin{theorem} \label{thm:LZBE_grammar_differ}
For an arbitrarily large $m$, there is a string $T$ with $|T| \in \Theta(m)$ such that $\gopt \in \Omega(\zbeg \alpha(\zbeg))$, where $\alpha(x)$ denotes the functional inverse of Ackermann's function.
\end{theorem}

To prove this, we consider the \emph{off-line range-product problem (ORPP)}.
\begin{definition}[\cite{DBLP:journals/ijcga/ChazelleR91}]
The \emph{off-line range-product problem (ORPP)} is defined as follows:

\noindent
{\bf Input:} a sequence $(a_1, \dots, a_m) \in \calX^m$, and a sequence of pairs of integers $(l_1, r_1), \dots, (l_{m'}, r_{m'})$ such that $1 \leq l_i \leq r_i \leq {m'}$ for all $i$.\\
{\bf Output:} a sequence $(q_1, \dots, q_{m'}) \in \calX^{m'}$ such that $q_i = \bigotimes_{j=l_i}^{r_i} a_j$ for each $i$.

Here, $\calX$ is a set and $\otimes$ is a binary operation such that $(\calX, \otimes)$ forms a semigroup.
\end{definition}
The following lower bound is known for this problem:
\begin{theorem}[\cite{DBLP:journals/ijcga/ChazelleR91}] \label{lem:ORPP_lowerbound}
For an arbitrarily large $m$, there exists an input to ORPP
that requires $\Omega(m \alpha(m))$ semigroup operations, where $m' \in O(m)$.
\end{theorem}

We construct an input string $T$ for ORPP over the alphabet $\Sigma = \{ x_1, \dots, x_m, \$_1, \dots, \$_{m+1} \}$ as follows:
\[
    T = x_1 \cdots x_m \$_1 Q_1 \$_2 Q_2 \$_3 \cdots \$_m Q_m \$_{m+1}
    \text{, where } Q_i = x_{l_i} \cdots x_{r_i}.
\]
We refer to each character $\$_i$ as a delimiter character, and let $\Sigma_x = \{x_1, \dots, x_m\}$ be the set of non-delimiter characters.  
The greedy LZBE factorization of $T$ consists of the first $m$ characters and each delimiter as a char factor, and each substring $x_{l_i} \cdots x_{r_i}$ as a copy factor referring to $[l_i, r_i]$.  
Thus, $\zbeg \in \Theta(m)$ holds.
We now show the following lemma.
\begin{lemma} \label{lem:grammar_ORPP}
If there exists a CFG $\calG$ of size $g$ producing $T$, the answer of ORPP can be computed by $O(g)$ semigroup operations.
\end{lemma}
\begin{proof}
We first convert the input CFG $\calG$ into an equivalent SLP $\calG'$ of size $O(g)$.  
Since $T$ contains $2m+1$ distinct characters, we have $g \in \Omega(m)$.
For any nonterminal $X$, let $\head{X}$ and $\tail{X}$ denote the longest prefix and suffix of $\exp{X}$ consisting of non-delimiter characters, respectively.  
For any string $C = x_{t_1} \cdots x_{t_{|C|}}$ over $\Sigma_x$, we define $\val{C} = \bigotimes_{j=1}^{|C|} x_{t_j}$.  
We compute $\val{\head{X}}$ and $\val{\tail{X}}$ for each nonterminal $X$ in a bottom-up fashion.  
For rules $X \rightarrow x_i$, the values are trivially $x_i$.  
For a nonterminal $X$ with a rule $X \rightarrow YZ$, the values are computed as follows:
\begin{align*}
    \val{\head{X}} &=
    \begin{cases}
        \val{\head{Y}} & \text{if } \exp{Y} \text{ contains a delimiter}, \\
        \val{\head{Y}} \otimes \val{\head{Z}} & \text{otherwise},
    \end{cases}\\
    \val{\tail{X}} &=
    \begin{cases}
        \val{\tail{Z}} & \text{if } \exp{Z} \text{ contains a delimiter}, \\
        \val{\tail{Y}} \otimes \val{\tail{Z}}\;\;\;\; & \text{otherwise}.
    \end{cases}
\end{align*}
An example of computing $\val{\head{X}}$ and $\val{\tail{X}}$ using the rules above is shown in Figure~\ref{fig:grammar_merge}.
Since the total number of nonterminals is $O(g)$, we can compute $\val{\head{X}}$ and $\val{\tail{X}}$ for all nonterminals using $O(g)$ semigroup operations.

Since $T$ contains delimiters $\$_1, \dots, \$_{m+1}$ in order, for each $1 \leq i \leq m$, there exists a unique rule $X \rightarrow YZ$ such that $Y$ contains $\$_i$ and $Z$ contains $\$_{i+1}$.
For such a rule, the string $Q_i$ corresponding to the $i$-th query range can be represented as $Q_i = \tail{Y} \cdot \head{Z}$.
Thus, the value $q_i$ of the $i$-th range-product query can be computed by
\[
q_i = \val{Q_i} = \val{\tail{Y} \cdot \head{Z}} = \val{\tail{Y}} \otimes \val{\head{Z}}.
\]
After computing $\val{\head{X}}$ and $\val{\tail{X}}$ for all $X$, we can compute all $q_i$ in $O(m)$ time.
Thus, the total number of operations required is $O(g + m) = O(g)$.
\end{proof}
We are now ready to prove Theorem \ref{thm:LZBE_grammar_differ}.
\begin{proof}[Proof of Theorem~\ref{thm:LZBE_grammar_differ}]
Let $m$ be an arbitrarily large integer, and consider an ORPP instance requiring $\Omega(m \alpha(m))$ semigroup operations by Lemma~\ref{lem:ORPP_lowerbound}.  
Construct the string $T$ from this instance as described earlier.  
The greedy LZBE factorization of $T$ has size $\zbeg \in \Theta(m)$.
Assume, for contradiction, that $\gopt \in o(\zbeg \alpha(\zbeg))$ for such a string $T$.  
Then, by Lemma~\ref{lem:grammar_ORPP}, the ORPP instance could be solved in $O(g)$ operations,  
contradicting the lower bound $\Omega(m \alpha(m)) = \Omega(\zbeg \alpha(\zbeg))$.  
Therefore, the theorem holds.
\end{proof}

Conversely, we can also prove that the above lower bound is tight even when 
considering the ratio between the smallest grammar size and the smallest LZBE factorization size.
This is shown by a reduction to the same problem together with an existing upper bound of ORPP.
\begin{theorem}[\cite{DBLP:journals/corr/abs-2406-06321}] \label{lem:ORPP_upperbound}
ORPP can be solved using $O(m \alpha(m))$ semigroup operations.
More precisely, for any input instance of ORPP, there exists a sequence $(b_1, \dots, b_t) \subseteq \calX^t$ of length $t \in O(m \alpha(m))$ such that 
\begin{enumerate}
    \item For each $1 \leq i \leq m$, $b_i = a_i$,
    \item For each $m+1 \leq i \leq t$, there exist integers $1 \leq c_i \leq d_i < i$ such that $b_i = b_{c_i} \otimes b_{d_i}$, and
    \item For each $1 \leq i \leq m$, there exists an integer $e_i$ with $b_{e_i} = q_i$.
\end{enumerate}
\end{theorem}
The construction given in the theorem directly yields an SLP of size $O(m \alpha(m))$, because the sequences $(b_1,\dots,b_t)$, $(c_{m+1}, \dots, c_t)$, and $(d_{m+1}, \dots, d_t)$ can be interpreted as a derivation sequence of nonterminals.
\begin{corollary} \label{cor:ORPP_upperbound_grammar}
For any input instance of ORPP, there exists an SLP of size 
$O(m\alpha(m))$ over the terminal alphabet 
$\Sigma_a = \{ a_1, \dots, a_m \}$ such that, 
for each $1 \leq i \leq m$, there exists a nonterminal $X_i$ 
satisfying $\exp{X_i} = a_{l_i} \cdots a_{r_i}$.
\end{corollary}

We now show the following theorem which directly implies $\gopt \in O(\zbeopt \alpha(\zbeopt))$.
\begin{theorem} \label{thm:ORPP_upperbound}
Given an LZBE factorization of a string $T$ with $\zbe$ factors,
we can construct a grammar deriving $T$ of size $O(\zbe \alpha(\zbe))$.
\end{theorem}
\begin{proof}
We begin by constructing an ORPP instance from the given LZBE factorization.
The instance consists of a sequence $(a_1, \dots, a_\zbe)$ and $\zbe$ queries, 
each corresponding to an LZBE factor.
If the $i$-th factor $F_i$ is a char factor, its query range is simply $(l_i, r_i) = (i, i)$.
If $F_i$ is a copy factor referring to $F_j \cdots F_k$, its query range is defined as $(l_i, r_i) = (j, k)$.

Using Corollary~\ref{cor:ORPP_upperbound_grammar},  
we obtain an SLP of size $O(\zbe \alpha(\zbe))$.  
We then convert this SLP with terminals $\Sigma_a$ into a CFG with terminals $\Sigma$ that generates the original string $T$.  
First, we introduce a new start symbol $S$ and add the rule $S \rightarrow a_1 \cdots a_\zbe$.  
Next, for each terminal $a_i$, we replace all its occurrences in the rules as follows.
\begin{enumerate}
    \item If $F_i$ is a char factor representing a character $c$, 
    then replace $a_i$ with a new nonterminal $X_i$ such that $\exp{X_i} = c$.
    \item If $F_i$ is a copy factor, 
    then replace $a_i$ with the existing nonterminal $X_i$ such that $\exp{X_i} = a_{l_i} \cdots a_{r_i}$.
\end{enumerate}
The existence of such a nonterminal
in the second case is guaranteed by Corollary \ref{cor:ORPP_upperbound_grammar}.
We can easily see that the size of the new grammar is $O(\zbe \alpha(\zbe))$.

Since the start symbol $S$ is expanded as $S \rightarrow a_1 \cdots a_\zbe$,  
it suffices to verify that each nonterminal $X_i$ replacing $a_i$ satisfies $\exp{X_i} = F_i$.  
We prove this by induction on $i$.
If $F_i$ is a char factor representing $c$, then by the definition of $X_i$, the nonterminal $X_i$ generates $c$, and hence $\exp{X_i} = F_i$.  
Now suppose that $F_i$ is a copy factor, and assume as the induction hypothesis that $\exp{X_{i'}} = F_{i'}$ for all $1 \le i' < i$.  
The new nonterminal $X_i$ corresponding to $F_i$ is defined by the rule $X_i \rightarrow X_{l_i} \cdots X_{r_i}$.  
By the induction hypothesis and the fact that $l_i < r_i < i$, each $X_{i'}$ with $l_i \le i' \le r_i$ expands to $F_{i'}$.  
Therefore, $\exp{X_i} = \exp{X_{l_i}} \cdots \exp{X_{r_i}} = F_{l_i} \cdots F_{r_i} = F_i$.
By induction, we have $\exp{X_i} = F_i$ for all $1 \le i \le \zbe$, and consequently $\exp{S} = F_1 \cdots F_\zbe = T$.  
Since the size of the constructed grammar is $O(\zbe \alpha(\zbe))$, the theorem follows.
\end{proof}

\begin{figure}[htbp]
    \centering
    \begin{minipage}[b]{0.05\linewidth}
    \end{minipage}
    \hfill
    \begin{minipage}[b]{0.4\linewidth}
        \centering
        \includegraphics[width=\linewidth]{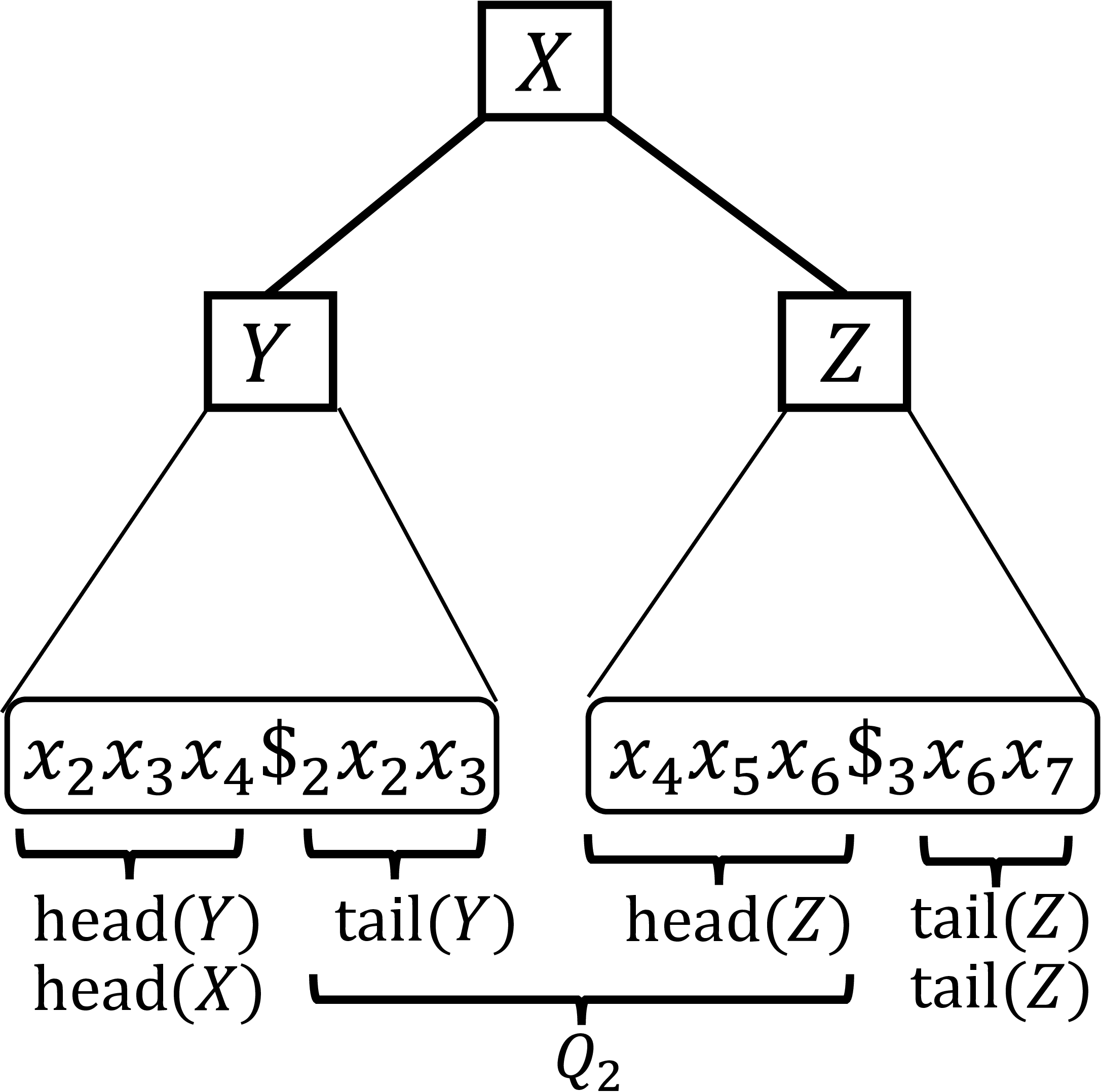}
    \end{minipage}
    \hfill
    \begin{minipage}[b]{0.4\linewidth}
        \centering
        \includegraphics[width=\linewidth]{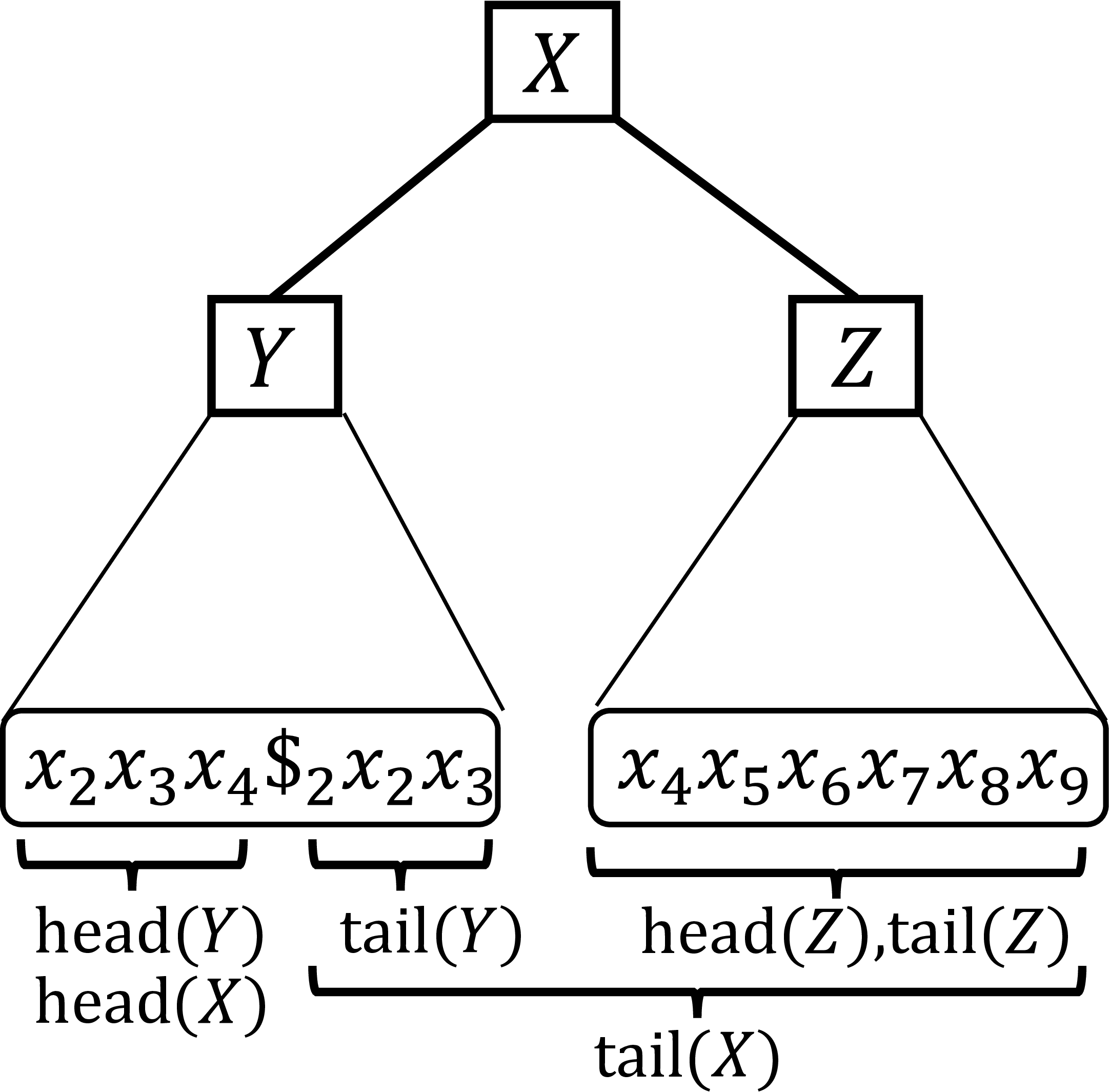}
    \end{minipage}
    \hfill
    \begin{minipage}[b]{0.05\linewidth}
    \end{minipage}
    \caption{
    Examples of bottom-up computation of $\val{\head{X}}$ and $\val{\tail{X}}$.  
    In the left example, the value $q_2$ is computed as $q_2 = \val{Q_2} = \val{\tail{Y} \cdot \head{Z}} = \val{\tail{Y}} \otimes \val{\head{Z}}$.
    }
    \label{fig:grammar_merge}
\end{figure}

We note that the lower bound also applies to run-length grammars~\cite{DBLP:conf/mfcs/NishimotoIIBT16}, which are an extension of CFGs allowing for run-length rules of the form $X \rightarrow Y^k$.  
We denote the size of the smallest run-length grammar generating $T$ by $\grlopt$.
Since the string $T$ constructed from the ORPP input is square-free, $\grlopt = \gopt$ holds.
Thus, for such $T$, we have $\grlopt \in \Omega(\zbeg \alpha(\zbeg))$.
This bound is tight even for optimal LZBE factorization; 
$\grlopt \in O(\zbeopt \alpha(\zbeopt))$ also holds for any $T$ since $\grlopt \leq \gopt$ and $\gopt \in O(\zbeopt \alpha(\zbeopt))$.

Conversely, there exists a string for which $\grlopt$ is asymptotically smaller than $\zbeopt$.
For example, the string $T = \mathtt{a}^n$ satisfies $\grlopt = O(1)$ and $\zbeopt = O(\log n)$,  
which implies that $\zbeopt \in \Omega(\grlopt \log n)$ in this case.  
This bound is tight, since $\gopt \in O(\grlopt \log n)$ and $\zbeopt \le \gopt$ together yield $\zbeopt \in O(\grlopt \log n)$ for any $T$.

\section{Data Structure for Efficient Random Access}

In this section, we introduce the following theorem.
\begin{theorem} \label{thm:LZBE_random_access}
Given an LZBE factorization $\calF = (F_1, \dots, F_\zbe)$ of a string $T$ of length $n$, 
we can construct, in $O(\zbe)$ time and space, a data structure of size $O(\zbe)$ that supports random access to any position of $T$ in $O(\log n)$ time.
\end{theorem}

An overview of our data structure is as follows.
We employ two central techniques: the \emph{interval-biased search tree (IBST)}~\cite{DBLP:journals/siamcomp/BilleLRSSW15} that is a variant of biased search trees~\cite{DBLP:journals/siamcomp/BentST85}, and \emph{symmetric centroid path decomposition (SymCPD)}~\cite{DBLP:journals/jacm/GanardiJL21} that is a generalization of the well-known heavy path decomposition~\cite{DBLP:journals/jcss/SleatorT83} for general DAGs.
We construct a global IBST over all factors in $\calF$ to simulate jump functions.
Additionally, we apply SymCPD to the dependency DAG $\calD$ of an LZBE factorization $\calF$, which yields a collection of \emph{heavy paths} with total length $O(\zbe)$.
For each heavy path, we build a separate IBST to accelerate the traversal along that path.

Throughout this section, we assume the standard RAM model.

\subsection{Building Blocks of the Data Structure}

\subsubsection{Interval-Biased Search Trees}
The interval-biased search tree (IBST)~\cite{DBLP:journals/siamcomp/BilleLRSSW15} is a weighted binary search tree that stores a sequence of consecutive disjoint intervals.
\begin{definition}[{Interval-Biased Search Tree (IBST)~\cite[Section 3]{DBLP:journals/siamcomp/BilleLRSSW15}}]
An interval-biased search tree is a binary tree that stores a collection of disjoint intervals
\[
\calI = \{ [a_0, a_1), [a_1, a_2), \dots, [a_{m-1}, a_m) \},
\]
where $a_0 < a_1 < \dots < a_m$ are strictly increasing integers.
Each node in the tree is associated with a single interval from $\calI$.
The tree is defined recursively as follows:
\begin{enumerate}
    \item The root stores the interval $[a_i, a_{i+1})$ such that the midpoint $(a_m - a_0) / 2$ lies within this interval, i.e., $(a_m - a_0)/2 \in [a_i, a_{i+1})$.
    \item The left child is the IBST constructed over the intervals $[a_0, a_1), \dots, [a_{i-1}, a_i)$, and the right child is the IBST constructed over $[a_{i+1}, a_{i+2}), \dots, [a_{m-1}, a_m)$.
\end{enumerate}
\end{definition}
We denote $v_i$ as the node storing $[a_i, a_{i+1})$.

Given a query value $q \in [a_0, a_m)$, we locate the interval in $\calI$ containing $q$ by traversing the IBST from the root.
At each node $v_i$, if $q \in [a_i, a_{i+1})$, we return it; otherwise, we move to the left child if $q < a_i$, or to the right child if $q \geq a_{i+1}$.  
Since each child subtree has at most half the total length of its parent, reaching a node $v_j$ from $v_i$ takes $O\left(\log \frac{L_i}{L_j}\right)$ steps, where $L_k$ is the total length of intervals in the subtree rooted at $v_k$.

Suppose that two boundaries $a_i$ and $a_j$ satisfying $a_i \leq q < a_j$ are known before querying $q$.
In this case, it is not necessary to start the search from the root of the IBST.  
Instead, we can efficiently locate the interval containing $q$ by exploring only the subtrees rooted at a few specific nodes.
We formalize this observation in the following lemma.
\begin{lemma} \label{lem:IBST_intervalsearch}
Let $a_i < a_j$ be any pair of interval boundaries.
By precomputing a constant number of nodes, referred to as \emph{hints}, associated with the interval $[a_i, a_j)$,
we can locate the interval $[a_x, a_{x+1}) \in \calI$ that contains a query value $q \in [a_i, a_j)$
in $O\left( \log \frac{a_j - a_i}{a_{x+1} - a_x} \right)$ time.
Furthermore, we can precompute hints for $h$ intervals in $O(m + h)$ time.
\end{lemma}
\begin{proof}
Let $\lca{v_x, v_y}$ denote the lowest common ancestor of nodes $v_x$ and $v_y$. 
We precompute the three hint nodes $v_c = \lca{v_i, v_{j-1}}$, $v_l = \lca{v_i, v_{c-1}}$, and $v_r = \lca{v_{c+1}, v_{j-1}}$. 
These three nodes are sufficient to perform the search without starting from the root of the IBST.

The algorithm first checks whether $q \in [a_c, a_{c+1})$. 
If so, we have found the target interval. 
Otherwise, if $q \in [a_l, a_c)$, we traverse the subtree rooted at $v_l$. 
If $q \in [a_{c+1}, a_r)$, we traverse the subtree rooted at $v_r$, similar to a standard top-down search. 
The initial check takes constant time, so we focus on the time complexity for traversing under $v_l$ and $v_r$. 
Consider the case where $q \in [a_l, a_c)$ and we traverse the subtree of $v_l$. 
Let $a_s$ be the starting boundary of the interval covered by the subtree rooted at $v_l$. 
Then $v_l$ covers $[a_s, a_c)$, its left subtree covers $[a_s, a_l)$, and its right subtree covers $[a_{l+1}, a_c)$ (See Figure~\ref{fig:IBST}). 
By the definition of IBST, the interval $[a_l, a_{l+1})$ stored by $v_l$ contains the midpoint $(a_c - a_s)/2$ of the total range.
Assume, for contradiction, that $(a_c + a_s)/2 < a_l$. 
Then $v_i$ lie in the left subtree of $v_l$, contradicting the assumption that $v_l$ is the lowest common ancestor of $v_i$ and $v_{c-1}$. 
Therefore we must have $a_l \leq (a_c + a_s)/2$, which implies $2a_l - a_c \leq a_s$. 
Since $a_i \leq a_l \leq a_c \leq a_j$, the length of the subtree is $O(a_j - a_i)$.
We start the traversal within an interval of length $O(a_j - a_i)$ and end in an interval of length $a_{x+1} - a_x$, 
Thus, the traversal time from $v_l$ is then $O(\log \frac{a_j - a_i}{a_{x+1} - a_x})$.
An analogous argument applies when traversing under $v_r$.
Since the off-line lowest common ancestor problem for a tree with $N$ nodes and $Q$ queries can be solved in $O(N + Q)$ time~\cite{DBLP:conf/focs/Harel80}, we can precompute hints for $h$ intervals in $O(m + h)$ time.
\end{proof}
Given interval boundaries $a_1, \dots, a_m$, we can construct the IBST in $O(m)$ time~\cite{DBLP:journals/siamcomp/BilleLRSSW15}. 

\begin{figure}[t]
    \centering
        \centering
        \includegraphics[width=0.8\linewidth]{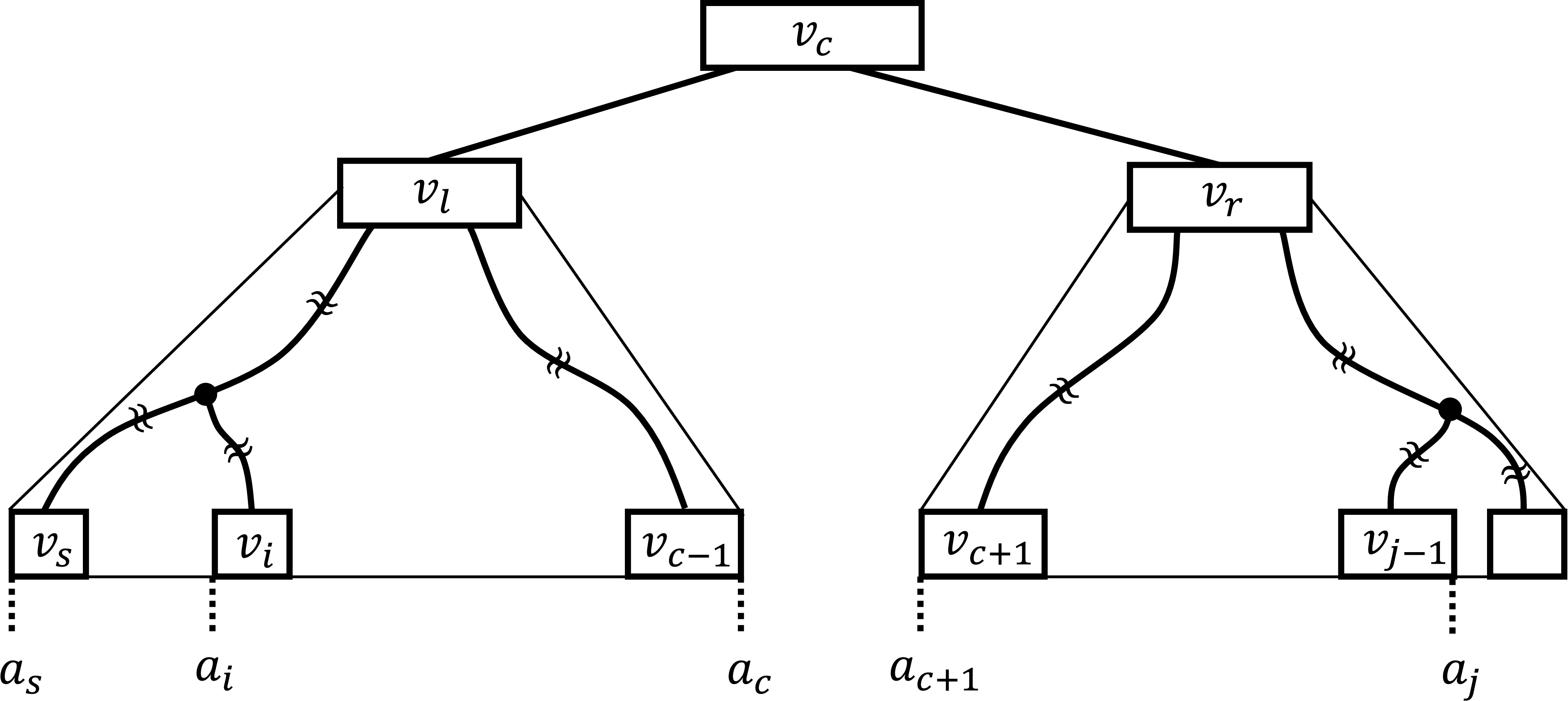}
    \caption{
    The subtree of the interval-biased search tree rooted at the node $v_c = \lca{v_i, v_{j-1}}$,  
    used to locate the interval containing a query value $q \in [a_i, a_j)$.
    The hint nodes $v_c$, $v_l = \lca{v_i, v_{c-1}}$, and $v_r = \lca{v_{c+1}, v_{j-1}}$ allow us to restrict the search to subtrees of limited range.
    }
    \label{fig:IBST}
\end{figure}

\subsubsection{Symmetric Centroid Path Decomposition}
The symmetric centroid path decomposition (SymCPD)~\cite{DBLP:journals/jacm/GanardiJL21} decomposes a DAG into disjoint paths, extending heavy path decomposition~\cite{DBLP:journals/jcss/SleatorT83} to arbitrary DAGs.
\begin{lemma}[{Symmetric Centroid Path Decomposition (SymCPD)~\cite[Lemma 2.1]{DBLP:journals/jacm/GanardiJL21}}] \label{lem:SymCPD}
For a DAG $\calD = \langle V, E \rangle$, 
there exists a pair of disjoint edge sets $H, L \subseteq E$ such that $L = E - H$, satisfying the following conditions:
\begin{enumerate}
    \item The edge-induced subgraph $\langle V, H \rangle$ consists of a set of disjoint paths.
    \item Any path in the DAG $\calD$ contains at most $2 \lg n_\calD$ edges from $L$,
\end{enumerate}
where $n_\calD$ denotes the number of maximal paths in $\calD$.
\end{lemma}
An edge in $H$ is called a \emph{heavy edge}, and a maximal path (including a path consisting of a single node) in the subgraph $\langle V, H \rangle$ is called a \emph{heavy path}. 
An edge in $L$ is called a \emph{light edge}.
Note that the total length of all heavy paths is bounded by $O(|V|)$, since each vertex has at most one outgoing heavy edge.

We apply this technique to the dependency DAG $\calD$ of an LZBE factorization $\calF$.
Since each maximal path corresponds to a jump sequence and the total number of jump sequences is $n$, we have $n_\calD \leq n$, and thus any path in $\calD$ contains at most $2 \lg n$ light edges.  
The next lemma shows that all heavy paths of the dependency DAG $\calD$ can be computed in $O(\zbe)$ time.

\begin{lemma} \label{lem:SymCPD_lineartime}
Given an LZBE factorization $\calF = (F_1, \dots, F_\zbe)$, all heavy paths in the dependency DAG $\calD$ can be computed in $O(\zbe)$ time.
\end{lemma}
\begin{proof}
The original symmetric centroid path decomposition selects an edge $(u, v)$ as a heavy edge if it satisfies the conditions $\lfloor \lg \ns{u} \rfloor = \lfloor \lg \ns{v} \rfloor$ and $\lfloor \lg \ne{u} \rfloor = \lfloor \lg \ne{v} \rfloor$, where $\ns{v}$ and $\ne{v}$ denote the number of maximal paths starting at and ending at node $v$, respectively.  
For any node $u$, the heavy edge among its outgoing edges points to the child with the largest value of $\ns{v}$.

Let $s_1, \dots, s_{\zbe}$ and $e_1, \dots, e_{\zbe}$ be arrays representing $\ns{F_i}$ and $\ne{F_i}$ for each factor $F_i \in \calF$.  
We initialize all entries to $0$, and compute the values as follows:
To compute $s_i$ (number of paths starting from $F_i$), process $i = 1, \dots, \zbe$ in increasing order:
\begin{itemize}
    \item If $F_i$ is a char factor, set $s_i \gets 1$.
    \item If $F_i$ is a copy factor referring to $F_j \cdots F_k$, set $s_i \gets \sum_{i'=j}^k s_{i'}$.
\end{itemize}
To compute $e_i$ (number of paths ending at $F_i$), process $i = \zbe, \dots, 1$ in decreasing order:
\begin{itemize}
    \item If $e_i = 0$, set $e_i \gets 1$.
    \item If $F_i$ is a copy factor referring to $F_j \cdots F_k$, update $e_{i'} \gets e_{i'} + e_i$ for all $i' = j$ to $k$.
\end{itemize}
Both procedures can be implemented in $O(\zbe)$ time using a prefix-sum technique.

Once the values of $s_i$ and $e_i$ are computed, we determine, for each copy factor $F_i$ referring to $F_j \cdots F_k$, the index $i' \in [j, k]$ that maximizes $s_{i'}$.
This corresponds to solving a range maximum query on each interval $[j, k]$.  
All such queries can be answered in total $O(\zbe)$ time by reducing them to offline lowest common ancestor (LCA) queries on a Cartesian tree~\cite{DBLP:journals/cacm/Vuillemin80} and applying a linear-time offline LCA algorithm~\cite{DBLP:conf/focs/Harel80}.
For each resulting candidate edge from $F_j$ to $F_i$, we then check whether it satisfies the heavy edge condition.
Since the number of such candidate edges is $O(\zbe)$, this step also runs in linear time.  
Thus, all heavy edges in the dependency DAG can be computed in $O(\zbe)$ time.  
Once all heavy edges are identified, the corresponding heavy paths can be constructed straightforwardly.
\end{proof}

\subsection{Our Data Structure}
In this subsection, we propose our data structure.
It mainly consists of two components: a global IBST to simulate the standard jump function, and heavy-path IBSTs to accelerate sequence of jump operations on the same heavy path.

We begin by simulating the jump function using an IBST.
Given an LZBE factorization $\calF$, 
we construct an IBST over the set of intervals $\{ [\posL{F}, \posR{F}] \mid F \in \calF \}$, which correspond to the positions of the factors in the output string.
For each $1 \leq i \leq \zbe$, we also precompute the hint nodes described in Lemma~\ref{lem:IBST_intervalsearch} for the interval $[\srcL{F_i}, \srcR{F_i}]$ to accelerate the search.
Since both the number of intervals and hints are $O(\zbe)$, the data structure can be stored in $O(\zbe)$ space.
Given a copy factor $F \in \calFcopy$ and a relative position $1 \leq r \leq |F|$, we can compute the jump function $\jump{F, r}$ by locating the interval that contains the position $\srcL{F} + r - 1$.
Using the hints for the interval $[\srcL{F}, \srcR{F}]$, we can find the corresponding factor $F'$ in $O\left(\log \frac{|F|}{|F'|}\right)$ time, where $F'$ is the next factor in the jump sequence.

We then focus on accelerating a sequence of jump operations on the same heavy path.
Let $P$ be a heavy path obtained by applying SymCPD to the dependency DAG~$\calD$ of an LZBE factorization.  
We denote the sequence of factors of $P$ as $(F_{i_1}, \dots, F_{i_{\ell}})$.
For any $1 \leq j < \ell$, we introduce two intervals $I_j^L$ and $I_j^R$, defined as follows:
\begin{itemize}
    \item $I_j^L = [\posL{F_{i_j}}, \posL{F_{i_j}} + \posL{F_{i_{j+1}}} - \srcL{F_{i_j}} - 1]$,
    \item $I_j^R = [
    \posR{F_{i_j}} 
    + \posR{F_{i_{j+1}}}
    - \srcR{F_{i_j}}
    + 1, \posR{F_{i_j}}]$.
\end{itemize}
For any $p \in [\posL{F_{i_j}}, \posR{F_{i_j}}] \setminus (I_j^L \cup I_j^R)$, the position $p$ refers to the next factor on the heavy path; that is, the factor corresponding to $\jump{\rel{p}}$ is $F_{i_{j+1}}$.
Conversely, for any $p \in I_j^L \cup I_j^R$, the position $p$ refers to the other factor.

Consider a factor $F_{i_s}$ on the heavy path $P$, a relative position $1 \leq r_s \leq |F_{i_s}|$, and the jump sequence starting from $(F_{i_s}, r_s)$.  
The beginning of the jump sequence takes the form $(F_{i_s}, r_s), \dots, (F_{i_e}, r_e)$, where each factor is in the same heavy path $P$.
The sequence terminates either when $(F_{i_e}, r_e)$ is the last element (i.e., $e = \ell$), or when $\jump{F_{i_e}, r_e}$ does not belong to the heavy path $P$.  
In the former case, the absolute position $\abs{F_{i_e}, r_e}$ lies within $[\posL{F_{i_e}}, \posR{F_{i_e}}]$, and in the latter case, it lies within either $I_e^L$ or $I_e^R$.  
We refer to this interval as the \emph{exit interval}, and $(F_{i_e}, r_e)$ as the \emph{exit position} of $(F_{i_s}, r_s)$.
Figure~\ref{fig:heavy_skip} illustrates an example of a heavy path and the corresponding exit intervals.

\begin{figure}[t]
    \centering
    \begin{minipage}[b]{0.60\linewidth}
        \centering
        \includegraphics[width=\linewidth]{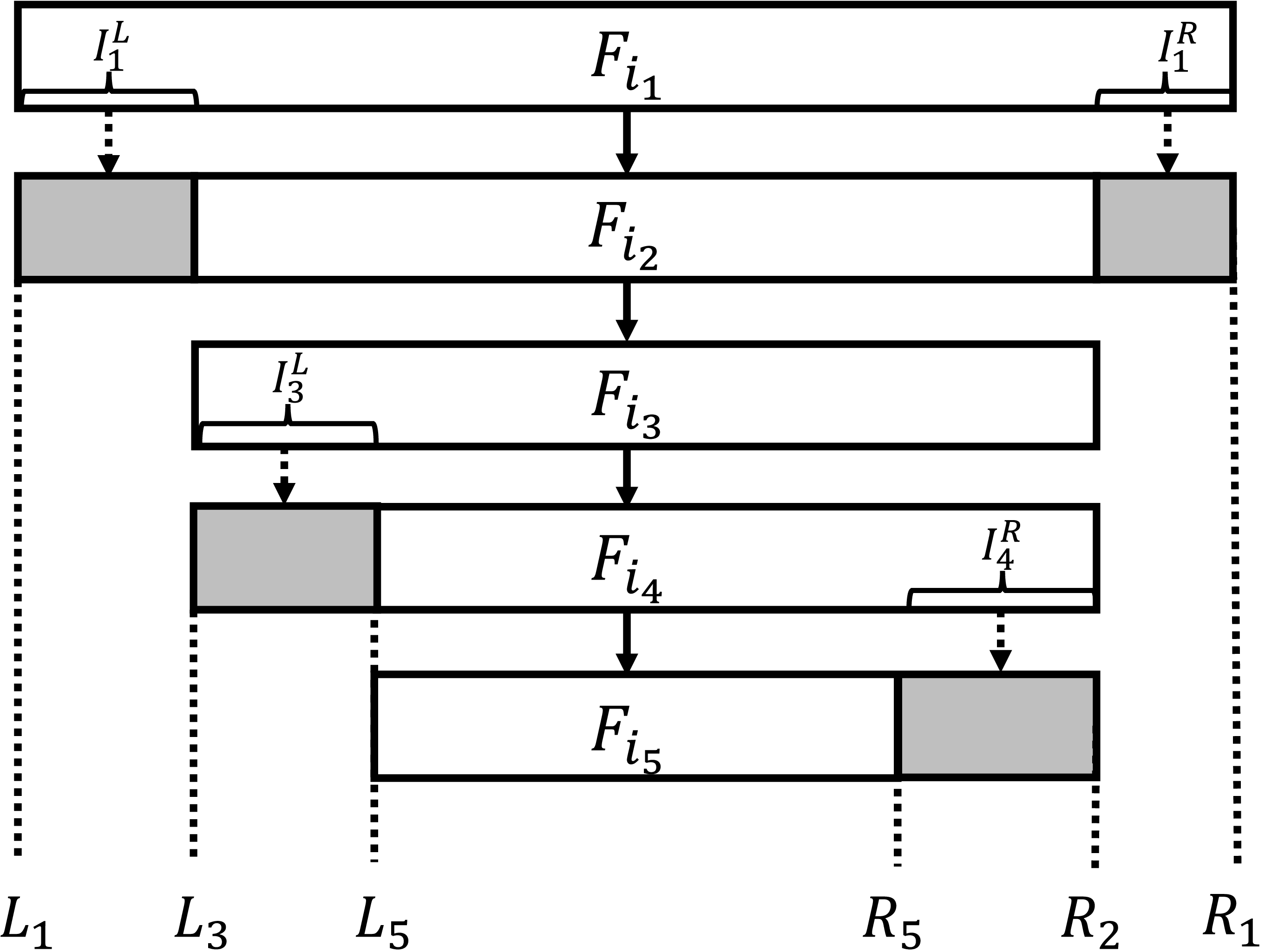}
    \end{minipage}
    \caption{
    An illustration of a heavy path in the dependency DAG.
    The path consists of five nodes $F_{i_1}$ through $F_{i_5}$, with solid arrows for heavy edges and dotted arrows for light edges.  
    Factors not belonging to this heavy path are shown in gray.
    The jump sequence from $(F_{i_1}, r_1)$, where $r_1 \in [R_5, R_2)$, has exit interval $I_4^R$ and exit position $(F_{i_4}, r_4)$, with $r_4 = R_5 - r_1$.
    }
    \label{fig:heavy_skip}
\end{figure}

We now present the following lemma.
\begin{lemma} \label{lem:heavy_skip}
For a heavy path $P = (F_{i_1}, \dots, F_{i_\ell})$,
there exists a data structure of size $O(\ell)$  
that, given a factor $F_{i_s}$ and a relative position $1 \leq r_s \leq |F_{i_s}|$,  
computes both the exit position $(F_{i_e}, r_e)$ and the corresponding exit interval $I$ in time $O\left( \log \frac{|F_{i_s}|}{|I|} \right)$.
\end{lemma}
\begin{proof}
For $1 \leq j \leq \ell$, we define $L_j$ and $R_j$ as follows:
\begin{align*}
    L_j &= \begin{cases}
        0 & \text{if } j = 1, \\
        L_{j-1} + |I_{j-1}^L| & \text{otherwise},
    \end{cases} \\
    R_j &= \begin{cases}
        L_j + |F_{i_j}| & \text{if } j = \ell, \\
        R_{j+1} + |I_j^R| & \text{otherwise}.
    \end{cases}
\end{align*}
Note that $R_j = L_j + |F_{i_j}|$ holds for any $j$ and the sequence $L_1, \dots, L_\ell, R_\ell, \dots, R_1$ is monotonically non-decreasing.
Using $L_j$ and $R_j$, we construct a set of disjoint consecutive intervals
$\calI'=\{ [L_1, L_2), \dots, [L_{\ell-1}, L_\ell), [L_\ell, R_\ell), [R_\ell, R_{\ell-1}), \dots, [R_2, R_1) \}$.

Given a query position $(F_{i_s}, r_s)$, we define $q = L_s + r_s - 1$ and find the interval in $\calI'$ that contains $q$.
Let $[L_j, L_{j+1})$, $[L_j, R_j)$, or $[R_{j+1}, R_j)$ be the matching interval.
In all cases, we have $j = e$, and the relative position is given by $r_e = q - L_j + 1 = L_s - L_j + r_s$.  
The corresponding exit interval is determined as follows:
\begin{itemize}
    \item If $q \in [L_j, L_{j+1})$, then the exit interval is $I_j^L$.
    \item If $q \in [R_{j+1}, R_j)$, then the exit interval is $I_j^R$.
    \item If $q \in [L_j, R_j)$ (in which case $j = e = \ell$), then the exit interval is $[\posL{F_{i_\ell}}, \posR{F_{i_\ell}}]$.
\end{itemize}
To find the interval containing $q$, we construct an IBST over $\calI'$ (excluding empty intervals).
Additionally, for each $1 \leq j \leq \ell$, we precompute the hint nodes for the interval $[L_j, R_j)$ to accelerate the search.
Given a query $(F_{i_s}, r_s)$, we identify the containing interval using the hints.
The length of the target interval corresponds to that of the exit intervals.
Thus, by Lemma~\ref{lem:IBST_intervalsearch}, we can compute the exit position in the claimed time complexity.
\end{proof}
For each heavy path in $\calF$, we construct the data structure to support efficient jump operations.

We combine these two data structures to support random access in $O(\log n)$ time.  
Our final structure consists of the global IBST $\calB_{\calF}$ over all factors in $\calF$ and a local IBST $\calB_P$ for each heavy path $P$. 
In addition, for each factor, we store its associated heavy path and its position within the path to facilitate efficient traversal.
We also precompute hint nodes for all IBSTs.  
For the global IBST $\calB_\calF$, we prepare hints for each interval $[\posL{F}, \posR{F}]$ of every factor $F$ and for all exit intervals.  
For each heavy-path IBST $\calB_P$, we precompute hints for the intervals $[L_i, R_i]$ where $1 \leq i < \ell$.

\begin{algorithm}[t]
    \caption{Algorithm for computing $T[p]$.}
    \label{alg:faster_random_access}
    \begin{algorithmic}[1]
        \State Compute $(F_1, r_1) \gets \rel{p}$ using $\calB_\calF$
        \State $t \gets 1$
        \While{$F_t$ is a copy factor}
            \State Let $P$ be the heavy path containing $F_t$
            \State Compute the exit position $(F_t', r_t')$ and the exit interval $I_t$ of $(F_t, r_t)$ using $\calB_P$ \label{lin:heavy_skip}
            \State Compute $(F_{t+1}, r_{t+1}) \gets \jump{F_t', r_t'}$ using $\calB_\calF$ \label{lin:light_edge}
            \State $t \gets t + 1$
        \EndWhile \\
        \Return \text{the character represented by} $F_t$
    \end{algorithmic}
\end{algorithm}

The random access algorithm is shown in Algorithm~\ref{alg:faster_random_access}.  
Given a position $p$, we first locate the factor $F_1$ such that $p \in [\posL{F_1}, \posR{F_1}]$, and compute the relative offset $r_1 = p - \posL{F_1} + 1$ using $\calB_{\calF}$.  
In each iteration, we use the IBST $\calB_P$ for the heavy path containing $F_t$ to compute the exit position $(F_t', r_t')$ and the corresponding exit interval $I_t$, using the precomputed hints for $F_t$.
We then determine the next jump target $(F_{t+1}, r_{t+1})$ using $\calB_{\calF}$ with the hints for $I_t$.  
The process continues until a char factor is reached.

Each IBST query after the first uses a precomputed hint and runs in $O(\log(X/Y))$ time,  
where $X$ is the length of the hint interval and $Y$ is the length of the interval found in the current query.  
By construction, the hint interval for each query is set to match the length of the interval found in the previous query.  
Thus, each query takes $O\left(\log \frac{\ell_t}{\ell_{t+1}}\right)$ time, where $\ell_t$ is the length of the interval found at the $t$-th query.  
Due to the structure of LZBE, each jump transitions to a factor no longer than the current one.  
Consequently, the sequence of interval lengths $\ell_1, \dots, \ell_{k+1}$ is non-increasing, where $k$ is the total number of jumps.
Furthermore, since a jump sequence corresponds to a path in the dependency DAG, $k$ is bounded by the number of heavy paths and light edges traversed along this path.  
By properties of heavy path decomposition, this gives $k = O(\log n)$.
Therefore, applying the telescoping sum technique, the total query time is bounded by
$\sum_{t=1}^{k} O\left(\log \frac{\ell_t}{\ell_{t+1}}\right) = O\left(k + \log \frac{\ell_1}{\ell_{k+1}}\right) = O(\log n)$.
The total number of intervals and hints used in both the global IBST and the heavy path IBSTs is $O(\zbe)$.  
Hence, the overall space complexity is $O(\zbe)$, and all components can be constructed in $O(\zbe)$ time.  
Since the mapping between factors and heavy paths can also be computed within the same bounds,  
the resulting data structure satisfies Theorem~\ref{thm:LZBE_random_access}.

\section{Linear-time Computation of the Greedy LZBE Factorization} \label{se:longest_greedy}
In this section, we present the following theorem.
\begin{theorem}
Given a string $T$ of length $n$, the greedy LZBE factorization of $T$ can be computed in $O(n)$ time.
\end{theorem}
We assume the word-RAM model with word size $\Omega(\log n)$, and a linearly-sortable alphabet.

Given a string $T$ and its greedy LZBE factorization of a prefix $T[1, p] = F_1 \cdots
F_k$, 
we define the \emph{extended factors} $E_1, \dots, E_k$ as follows:
\[
E_i = \begin{cases}
F_i F_{i+1} & \text{if } F_i \in \{ E_1, \dots, E_{i-1} \}, \\
F_i & \text{otherwise}.
\end{cases}
\]
If the last factor $E_k$ is equal to some earlier extended factor (i.e., $E_k \in \{ E_1, \dots, E_{k-1} \}$), we do not define $E_k$.
We begin by establishing several combinatorial properties of the extended factors that arise in the greedy LZBE factorization.
\begin{lemma} \label{lem:extended_factor_distinct}
No string appears more than twice as an extended factor.
\end{lemma}
\begin{proof}
We prove this lemma by showing that any two non-consecutive extended factors (i.e., those with $|i - j| \geq 2$) must be distinct.  
Assume for contradiction that $E_i = E_j$ for some $i + 1 < j$.
If $E_j = F_j$, then $E_i = E_j$ implies that $E_j$ appears in $\{ E_1, \dots, E_{j-1} \}$,
contradicting the definition of extended factors.  
Therefore, it must be that $E_j = F_j F_{j+1}$.  
In that case, since $E_i = E_j = F_j F_{j+1}$ appears earlier than $\posL{F_j}$ as a consecutive sequence of factors, choosing $F_j$ as the next factor violates the longest-greedy condition, as $E_i$ would have been the longer valid choice.  
Thus, non-consecutive extended factors must be distinct.
\end{proof}

\begin{lemma}
Let $F_1 \cdots F_k$ be the greedy LZBE factorization of $T[1, p]$.  
If $F_{k+1}$ is a copy factor and its leftmost valid source occurrence is $F_i \cdots F_j$, then $F_{k+1}$ begins with $E_i$.
\end{lemma}
\begin{proof}
If $i < j$, then $F_{k+1}$ begins with $F_i F_{i+1} = E_i$, so the claim holds.
If $E_i = F_i$, the claim also clearly holds.  
It remains to consider the case where $E_i = F_i F_{i+1}$ and $F_{k+1} = F_i$.
In this case, $F_i$ already appears as an earlier extended factor.  
Thus, $F_i$ would not be chosen as the start of a new factor since the algorithm selects the leftmost valid occurrence.
Therefore, this case does not occur, and the claim holds.
\end{proof}

We begin by describing a simple algorithm that does not provide strong theoretical guarantees.
This algorithm maintains a trie representing all extended factors computed so far.
Suppose we have already computed the factorization $F_1 \cdots F_k$ for the prefix $T[1, p]$.
To determine the next factor starting at position $p+1$, we follow the path in the trie along the suffix $T[p+1, n]$.
For each extended factor $E_i$ on this path, we attempt to extend $E_i$ to the longest
prefix of $T[p+1, n]$ that can be expressed as a concatenation $F_i \cdots F_j$ for some $i \le j$.
This is done by computing the longest common prefix (LCP) length $d$ between $T[p+1, n]$ and $T[\posL{F_i}, n]$, and locating the factor containing the position $\posL{F_i} + d$.
Each such computation yields a candidate factor $F_i \cdots F_j$, and we select the longest among
them as $F_{k+1}$.
We can answer LCP queries in constant time by using an LCP array~\cite{DBLP:journals/siamcomp/ManberM93} together with a range minimum data structure~\cite{DBLP:conf/latin/BenderF00}, both of which can be constructed in $O(n)$ time and space.

Although conceptually simple, this approach offers no strong time bound.
In the worst case, the path traversal for computing $F_{k+1}$ may traverse far beyond the length of the factor eventually selected.
A straightforward analysis shows that the running time becomes $O(n^2)$, and it is unclear whether a faster
guarantee can be obtained for this method.
We therefore turn to a refined approach that achieves an $O(n)$ running time.

To obtain a linear time algorithm, we replace the extended factor trie with the suffix tree ST~\cite{DBLP:conf/focs/Weiner73} of $T$ and augment it with additional structures.
We construct a Weighted Level Ancestor  (WLA) structure~\cite{DBLP:conf/cpm/BelazzouguiKPR21} and an incremental Nearest Marked Ancestor (NMA) structure~\cite{DBLP:conf/focs/AlstrupHR98} on the ST.
Together, these allow us to locate the locus of any substring and to maintain marks on the loci with amortized constant time updates and queries.
All components can be constructed in $O(n)$ time and stored in $O(n)$ space.
In the algorithm, each extended factor $E_i$ is represented by a mark at the corresponding locus.
With this refinement, enumerating candidate extended factors is replaced by walking up from the leaf representing the suffix $T[p+1, n]$ and repeatedly computing NMAs.
For each retrieved mark corresponding to an extended factor, we compute the LCP value between the suffix and extendend factor, and obtain a candidate of the next factor.
After selecting $F_{k+1}$, we insert marks for the newly created extended factors, and the NMA structure maintains them.

We now analyze the running time of this algorithm.
Consider computing $F_{k+1}$ starting from the position $p+1$, and let $d$ be the maximum depth of any marked locus on the path representing $T[p+1, n]$.
Since each locus can carry at most two marks by Lemma~\ref{lem:extended_factor_distinct}, the path contains at most $2d$ marks in total.
The deepest mark on the path corresponds to an extended factor that matches a prefix of $T[p+1, n]$ of length $d$, which implies that the length of the next factor $F_{k+1}$ is at least $d$.
Consequently, the number of candidate extended factors that must be examined is $O(|F_{k+1}|)$.
Because each LCP, WLA, NMA, and marking operation runs in amortized constant time, computing $F_{k+1}$ and updating the data structures takes amortized $O(|F_{k+1}|)$ time.
Since the total length of all factors is $n$, summing over all iterations shows that the greedy LZBE
factorization can be computed in overall $O(n)$ time.

\section{Conclusion}

We introduced \emph{LZ-Begin-End} (LZBE), a new LZ-like factorization where every copy factor refers to a contiguous sequence of preceding factors.
On the structural side, we showed that any CFG of size $g$ can be converted into an LZBE factorization of size at most $g$, implying $\zbeopt \le \gopt$.
We further established a separation in the other direction, proving that there exist strings for which $\gopt \in \Omega(\zbeg \alpha(\zbeg))$.
Complementing this lower bound, we gave an $O(\zbeopt \alpha(\zbeopt))$ upper bound on the grammar size, and analogous statements for run length grammars.
On the algorithmic side, we presented a linear time algorithm that computes the greedy LZBE factorization.
We also designed a random access index that supports access in $O(\log n)$ time using $O(\zbe)$ space and build time, by combining interval biased search trees with symmetric centroid path decomposition over the dependency DAG.

\myblock{Acknowledgments}
This work was supported by JSPS KAKENHI Grant Numbers JP25K00136 (YN), JP20K19743, JP20H00605, JP25H01114 (YY), JP23K18466, JP23K24808 (SI), and JST CRONOS Japan Grant Number JPMJCS24K2 (YY).

\clearpage

\FloatBarrier

\bibliography{mybib}

\clearpage

\appendix

\section{A Lower Bound on the Approximation Ratio of Greedy to Optimal LZBE Factorization} \label{app:greedy_approx}
\begin{theorem}
There exists a family of binary strings such that the greedy LZBE factorization produces asymptotically twice as many factors as some alternative valid LZBE factorization.
\end{theorem}
\begin{proof}
For any positive integer $m > 1$, define the strings $A$, $B$, $A_i$, and $T_{i,j}$ over the alphabet $\{ a, b \}$ for $0 \leq i < m$ and $0 \leq j < m$ as follows:
\begin{align*}
    A &= a^{2^{m+1}}, \\
    B &= b^{2^{m+1}} b, \\
    A_i &= \prod_{k=m-i+1}^{m} a^{2^k},\\
    T_{i,j} &= A B 
        \cdot
        \left( 
        \prod_{i'=1}^{i-1} 
            \prod_{j'=1}^{m-1} 
                A_{i'} b^{2^{j'+1}} b
        \right)
        \cdot
        \left(
        \prod_{j'=1}^{j} 
            A_i b^{2^{j'+1}} b 
        \right),
\end{align*}
where $\prod_{i=x}^y S_i$ denotes the concatenation of $S_x, \dots, S_y$, and is defined as the empty string if $x > y$.
Note that $T_{i,j}$ can be obtained by appending $A_i b^{2^{j+1}} b$ to the end of $T_{i,j-1}$, and that $T_{i, m-1} = T_{i+1, 0}$ hold.
The string $T_{m-1,m-1}$ gives a lower bound.

We first analyze the greedy LZBE factorization of $T_{i,j}$ by induction.
The base case is $T_{0, 0} = A B$, which is parsed as:
\[
a|a|a^2|a^4|\dots|a^{2^m}|b|b|b^2|b^4|\dots|b^{2^m}|b.
\]
When constructing the factorization of $T_{i,j}$ from that of $T_{i,j-1}$, all factors of $T_{i,j-1}$ except the last one remain unchanged. We prove the following:
\begin{itemize}
    \item If $j = 1$, the greedy factorization of $T_{i,j}$ is obtained by appending two new factors: $A_i b^{2^{j+1}}$ and $b$.
    \item If $j > 1$, the greedy factorization of $T_{i,j}$ is obtained by replacing the last factor $b$ of $T_{i,j-1}$ with $b A_i b^{2^j}$, and then appending two new factors: $b^{2^j}$ and $b$.
\end{itemize}
From this, we observe that the last factor of $T_{i, j}$ is always $b$ for all $1 \leq i \leq m$ and $0 \leq j < m$.
We first consider the case $j = 1$.  
Here, $T_{i, j-1} = T_{i,0} = T_{i-1, m-1}$, and the appended string is $A_i b^4 b$.  
Since the length of $A_i$ increases with $i$, the substring $A_i$ does not appear in $T_{i,j-1}$, except possibly as part of $A$.  
Thus, the last factor $b$ of $T_{i,j-1}$ remains unchanged in $T_{i,j}$.  
We can take $A_i b^4$ as a concatenation of factors by referring to $A_i$ from the suffix of $A$ and $b^4$ from the prefix of $B$.  
However, we cannot take $A_i b^5$ as a factor, because the factorization of $B$ follows the structure $b|b|b^2|b^4|\dots$, and hence only substrings like $A_i b^{2^x}$ can be taken as consecutive factors.  
Thus, the factorization for this case is correct.

Now consider the case $j > 1$.  
The suffix of $T_{i,j}$ has the form:
\[
b \cdot A_i b^{2^j} b \cdot A_i b^{2^{j+1}} b.
\]
By the induction hypothesis, the previous two factors in the greedy factorization of $T_{i,j-1}$ are 
$b$ and $A_i b^{2^j}$ when $j=2$, and $b A_i b^{2^{j-1}}$ and $b^{2^{j-1}}$ when $j>2$.
Therefore, we now need to determine the factorization for the new portion $b A_i b^{2^{j+1}} b$.
In this case, we can take $b A_i b^{2^j}$ by referring to the two previous consecutive factors.  
However, we cannot take a longer match, as the substring $b A_i b^{2^j}$ appears nowhere else in the factorized part of $T_{i,j}$ except in this location.  
The remaining suffix $b^{2^j} b$ includes a factor $b^{2^j}$ from $B$, but we cannot take $b^{2^j} b$ as a concatenation of existing factors.  
Hence, the greedy factorization proceeds correctly in this case as well.
Finally, we count the total number of factors.  
The base case $T_{0, 0}$ contains $2m + 5$ factors.  
The difference in the number of factors between $T_{i, j-1}$ and $T_{i,j}$ is always 2.  
Thus, the total number of factors in the greedy factorization of $T_{m-1,m-1}$ is:
\[
2m + 5 + 2(m - 1)(m - 1) = 2m^2 - 2m + 7.
\]

We now give an LZBE factorization of $T_{m-1,m-1}$ that produces a smaller number of factors.  
We parse the prefix $T_{0, 0} = A B$ of $T_{m-1,m-1}$ as:
\[
a|a|a^2|a^4|\dots|a^{2^m}|b|b|b|b^2|b^4|\dots|b^{2^m}.
\]
The only difference from the greedy factorization lies in the part of $B$.  
This factorization allows substrings of the form $A_i b^{2^{j+1}} b$ to be represented as a concatenation of multiple existing factors for all $1 \leq i < m$ and $1 \leq j < m$.
Thus, we can construct the factorization of $T_{i,j}$ recursively by concatenating the factorization of $T_{i,j-1}$ with a single new factor $A_i b^{2^{j+1}} b$.
The total number of factors in this factorization of $T_{m-1,m-1}$ is:
\[
2m + 5 + (m - 1)(m - 1) = m^2 + 6.
\]
Therefore, as $m \rightarrow \infty$, the ratio between the number of factors in the greedy factorization and in this factorization approaches 2 (See Figure \ref{fig:LZBE_greedy_opt}).
\end{proof}

\begin{figure}[t]
    \centering
        \centering
        \includegraphics[width=0.9\linewidth]{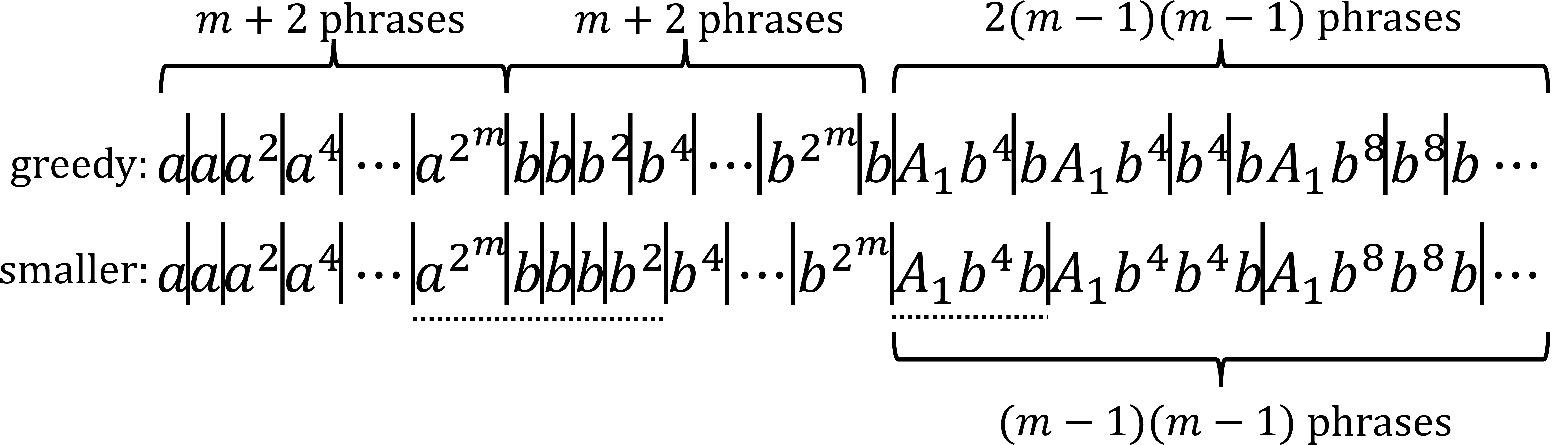}
    \caption{
    Illustration of the greedy LZBE factorization and a smaller LZBE factorization of the string $T_{m-1,m-1}$.  
    In the smaller factorization, the substring $A_1 b^4 b = a^{2^m} b^5$ (dotted lines) is selected as a single factor.
    }
    \label{fig:LZBE_greedy_opt}
\end{figure}
Although we have constructed a family of strings for which the ratio between the greedy LZBE factorization and a smaller one approaches 2,
we have not found any example where the ratio between the greedy and the optimal factorization exceeds 2.
This resembles the case of LZ-End, for which a 2-approximation is also conjectured.  
Motivated by this parallel, we conjecture that the greedy LZBE factorization achieves a 2-approximation; that is, for any string $T$, we have $\zbeg \leq 2 \zbeopt$.

\end{document}